\documentclass{article}
\pdfoutput = 1

\usepackage{amssymb, amsmath, amsthm}
\usepackage[utf8]{inputenc}
\usepackage{stmaryrd}
\usepackage[noend]{algpseudocode}
\usepackage{algorithmicx}
\usepackage{siunitx}
\usepackage{graphicx,graphics}

\title{Abstract Interpretation with Higher-Dimensional
Ellipsoids\\ and Conic Extrapolation
\thanks{Published in the Proceedings of CAV 2015.
The final publication is available at
{http://link.springer.com/chapter/10.1007/978-3-319-21690-4\_24} } }

\author{Mendes Oulamara\thanks{This material is based upon work
supported by the National Science Foundation under Grant
No.~1136008.}\\
\emph{École Normale Supérieure} \\
\emph{45 rue d'Ulm} \\ \emph{75005 Paris, France}\\
mendes.oulamara@ens.fr\\
 \and Arnaud J. Venet \\
 \emph{Carnegie Mellon University}\\
 \emph{NASA Ames Research Center}\\
\emph{Moffett Field, CA 94035} \\
arnaud.venet@west.cmu.edu}

\date{CAV 2015, 18-24 July 2015}

\theoremstyle{plain} \newtheorem{theorem}{Theorem}
\theoremstyle{plain} \newtheorem{lemma}{Lemma}
\theoremstyle{definition} \newtheorem{definition}{Definition}

\newcommand{\st}{\text{ s.t. }}
\newcommand{\W}{\bigtriangledown}
\newcommand{\Sk}{\sum\limits_{i=1}^k}
\newcommand{\segk}{\llbracket1,k\rrbracket}
\newcommand{\on}[1]{\operatorname{#1}}
\newcommand{\Ell}{\operatorname{Ell}}

\begin{document}
\maketitle

\begin{abstract}
The inference and the verification of numerical relationships among
variables of a program is one of the main goals of static analysis.
In this paper, we propose an Abstract Interpretation framework based
on higher-dimensional ellipsoids to automatically discover symbolic
quadratic invariants within loops, using loop counters as implicit
parameters. In order to obtain non-trivial invariants, the diameter
of the set of values taken by the numerical variables of the program
has to evolve (sub-)linearly during loop iterations. These invariants
are called ellipsoidal cones and can be seen as an extension of
constructs used in the static analysis of digital filters.
Semidefinite programming is used to both compute the numerical
results of the domain operations and provide proofs (witnesses)
of their correctness.

\textbf{keywords:} static analysis, semidefinite programming,
ellipsoids, conic extrapolation
\end{abstract}

\section{Introduction}

Ellipsoids have been widely used
to overapproximate convex sets. For instance, in Control Theory
they naturally arise as sublevel sets of quadratic Lyapunov
functions. They are chosen to minimize some criterion, such as the
volume.
In Abstract Interpretation~\cite{cousot}, they have been used to
compute
bounds on the output of linear digital filters~\cite{esop,nsad}.
Roux~\textit{et al.}~\cite{plg,plg2} further extended that approach
by borrowing
techniques from Semidefinite Programming (SDP).
However, all those works try to recover an ellipsoid that is known
to exist as
the Lyapunov invariant of some control system from the numerical
algorithm implementing that system.
The analysis algorithms are tailored for the particular type of
numerical code considered.
Ellipsoids are interesting in and of themselves because they
provide a space-efficient
yet expressive representation of convex sets in higher dimensions
(quadratic compared to exponential for polyhedra).
In this paper, we devise an Abstract Interpretation
framework~\cite{cousot2} to
automatically compute an overapproximation of the values of the
numerical variables in a program
by an ellipsoid.

We focus our attention on the case when the program variables grow linearly with
respect to the enclosing loop counters. We call this approximation
an `ellipsoidal cone'.
Our work also relates to the gauge domain
\cite{gauge}, which discovers simple linear
relations between loop counters and the numerical variables of a
program.
Even though the definitions of the abstract operations are general,
this model arises more naturally when the analyzed system naturally tends to
exhibit quadratic invariants, for instance in the analysis of switched
linear systems.
Section~\ref{ellop} defines the basic ellipsoidal operations and
their verification, and Sect.~\ref{conop} extends this to the
conic extrapolation.
The soundness of our analysis relies on the verification of Linear
Matrix Inequalities (LMI), which we describe in Sect.~\ref{LMI}
before delving into the description of ellipsoidal cones.
Finally Sect.~\ref{switch} presents experiments and discusses
applications to switched linear systems.

\section{Ellipsoidal Operations}
\label{ellop}

Ellipsoids are the building blocks of our conic extrapolation. We
define how to compute the result of basic operations (union,
affine transformation\ldots). Since there is generally no minimal
ellipsoid in the sense of inclusion, we choose the heuristic of
minimizing the volume. Other choices, such as minimizing the so
called `condition number' or preserving the shape, are compared in
\cite{plg}.

We mainly rely on SDP optimization methods
\cite{lectmodconv,sdp,sdp2} both to find a
covering ellipsoid and test the soundness of our result. However,
we do not rely on the correctness of the SDP solver. For each
operation whose arguments and results are expressed in function of
matrices $(A_i)_{1\le i\le r}$, we define a linear matrix inequality (LMI) of
the form
$
\sum\limits_{i=0}^r \alpha_iA_i \succeq 0
$,
where $A\succeq 0$ means ``$A$ is semidefinite positive'', such that
proving the soundness of the result is equivalent to showing that
the LMI is satisfied for some reals $(\alpha_i)$. We find
$(\alpha_i)$
candidates using an SDP solver and then verify the inequality with
a sound procedure described in Sect.~\ref{LMI}.

\begin{definition}[Ellipsoid]
\label{defell}
$\Ell(Q,c)=\{ x\in\mathbb{R}^n | (x-c)^TQ(x-c)\le 1\}$ is the definition
of an ellipsoid where
$c\in\mathbb{R}^n$ and $Q$ is a definite positive $n\times n$ matrix.
For practical use, we also define the function
$F:(Q,c)\mapsto
\begin{pmatrix}
Q & -Qc \\
-c^TQ & ~c^TQc-1
\end{pmatrix}$.
\end{definition}

\subsection{A Test of Inclusion}

Let $\Ell(Q,c)$ and $\Ell(Q^*,c^*)$ be two ellipsoids, using the
function $F$ of Definition~\ref{defell} we have the
following duality result (proven in \cite{yildrim}):
\begin{theorem}
\label{Sproc}

\begin{multline*}
\max_{x\in\Ell(Q,c)}\left((x-c^*)^TQ^*(x-c^*)-1\right) = \\
\min_{\lambda,\beta\in\mathbb{R}}\left\{\beta \st \lambda\ge 0 \mbox{ and }
\beta E_{n+1} + \lambda F(Q,c) \succeq F(Q^*,c^*)
\right\}
\end{multline*}

Where $E_{n+1}$ is an $(n+1)\times (n+1)$ matrix, with $(E_{n+1})_{i,j}=1$ if
$i=j=n+1$, else $(E_{n+1})_{i,j}=0$.
Hence $\Ell(Q,c)\subset \Ell(Q^*, c^*)$ if and only if the minimizing
value $\beta^*$ is nonpositive.
\end{theorem}

For given $\Ell(Q,c)$, $\Ell(Q^*,c^*)$ and candidates $\lambda$ and
$\beta$ computed by the SDP solver, the right hand term provides the
LMI to check.

\subsection{Computation of the Union}
Let $(\Ell(Q_i,c_i))_{1\le i\le p}$ be $p$ ellipsoids, we want to
find an ellipsoid $\Ell(Q^*, c^*)$ which is of nearly minimal volume
containing them.
To do so, we can solve the following SDP problem. It is decomposed into
a first part ensuring the inclusion, proven in \cite{yildrim}, and a
second part describing the volume minimization criterion, proven in
\cite[example 18d]{lectmodconv}.

The unknowns of the SDP problem are $X$ an $n\times n$
symmetric matrix, $z\in\mathbb{R}^n$ a vector, $\Delta$ a lower triangular matrix
and real numbers $t$, $(\tau_i)_{1\le i \le p}$ $(u_i)_{1\le i \le 2^{l+1}-2}$ where
$n\le 2^l < 2n$:

\begin{equation}
\begin{aligned}
\label{sdpU}
\mbox{\textbf{maximize} } &t \mbox{ \textbf{such that}}&\\
&\mbox{\textbf{Inclusion conditions, see \cite{yildrim}:}}\\
& \forall i, 1\le i\le p,\mbox{ } \exists \tau_i\ge 0 \st \\
&	\tau_i
	\begin{pmatrix}
	Q_i & ~-Q_ic_i & 0 \\
	-c_i^TQ_i & ~c_i^TQ_ic_i-1 ~~& 0 \\
	0 & 0 & 0
	\end{pmatrix}
	\succeq
	\begin{pmatrix}
	X & -z & 0 \\
	-z^T~ & -1~ & z^T \\
	0 & z & -X
	\end{pmatrix} \\
&\mbox{\textbf{
		  	Volume minimization, see \cite[example 18d]{lectmodconv}:} }\\
&\begin{pmatrix}
X & \Delta \\
\Delta^T & D(\Delta)
\end{pmatrix}
\succeq 0 \\
&\mbox{ where $D(\Delta)$ is the diagonal matrix with the
			 diagonal of $\Delta$. } \\
&\begin{pmatrix}
u_1 & t \\
t & u_2
\end{pmatrix} \succeq 0 \mbox{ and }
\forall i, 1\le i \le 2^l-2,
\begin{pmatrix}
u_{2i+1} & u_i \\
u_i & u_{2i+2}
\end{pmatrix} \succeq 0 \\
&\forall i, 2^l-1\le i < 2^l-1+n, u_i=\delta_{i-2^l+2} \\
&	\mbox{ where $(\delta_1,\ldots,\delta_n)$ are the diagonal
	coefficients of $\Delta$.} \\
&\forall i, 2^l-1+n\le i \le 2^{l+1}-2, u_i=1
\mbox{ }\\
\end{aligned}
\end{equation}

We then define $Q^*=X$ and $c^*=Q^{*-1}z$ (in floating-point
numbers, then we possibly increase the ratio of $Q$ to
ensure the inclusion condition).
We can check that the resulting ellipsoid really contains the others
with Theorem~\ref{Sproc}.

\subsection{Affine Assignments}
\label{ellaff}
In this section, we are interested in computing the sound
counterpart of an assignment
$
x \gets Ax+b
$, where $x$ is the vector of variables, $A$ is a matrix and $b$ a
vector.

\subsubsection{Computation.}
We want to find a minimal volume ellipsoid such that the inclusion~
$\Ell(Q^*, c^*)\supset \{Ax+b | (x-c)^TQ(x-c)\le 1\}$ is verified.

By a symmetry argument, we can set $c^*=Ac+b$. By expanding the inclusion
equation, we find ~
$
\{Ax+b | x\in\Ell(Q,c)\}
				\subset \Ell(Q^*,Ac+b)
\iff  Q\succeq A^TQ^*A
$.

Hence $Q^*$ is a solution of the following SDP problem with unknowns
$X$ an $n\times n$
symmetric matrix, $z\in\mathbb{R}^n$ a vector, $\Delta$ a lower triangular matrix
and real numbers $t$, $(\tau_i)_{1\le i \le p}$ $(u_i)_{1\le i \le 2^{l+1}-2}$ where
$n\le 2^l < 2n$:

\begin{equation}
\begin{aligned}
&\mbox{\textbf{Volume minimization:}}\\
	&\mbox{The same constraints and objective as in (\ref{sdpU}).}\\
&\mbox{\textbf{Inclusion conditions:}}\\
&Q \succeq A^TXA \mbox{ and }
	\frac{1}{\epsilon}\on{Id} \succeq X
	\mbox{ where $\on{Id}$ is the $n\times n$ identity matrix and $\epsilon>0$}
\end{aligned}
\end{equation}

We then set $Q^*=X$ and $c^*=Ac+b$.

The second inclusion condition is here to ensure the numerical
convergence of the SDP solving algorithm: if $A$ is singular, the
image by $A$ of an ellipsoid is a flat ellipsoid. With
this condition, we ensure that $\Ell(Q^*, c^*)$ contains a ball of
radius $\epsilon$.

To add an input defined by the convex hull of a finite set of vectors
(for instance a hypercube), we can just compute the sum for every
one of these vectors and compute the union~\cite{plg}.

\subsubsection{Verification.}
The previous procedure gives us inequalities whose correctness
ensures that the ellipsoid $\Ell(Q^*, Ac+b)$ contains the image of
$\Ell(Q,c)$ by
$x\mapsto Ax+b$. However, $c^*=Ac+b$ is computed in floating-point
arithmetic, hence the soundness does not extend to our actual result
$\Ell(Q^*,c^*)$. Therefore, we have to devise a test of inclusion for
an arbitrary $c^*$. Let us compute the resulting center in two steps:
 we first assume that $b=0$. We have from \cite{yildrim}:

\begin{equation}
\label{verifassign}
\begin{aligned}
&\left(\forall x, x\in\Ell(Q,c)
		\Rightarrow Ax\in\Ell(Q^*,c^*)\right)\\
\iff &
\max_{x\in\Ell(Q,c)}(x^TA^TQ^*Ax-2x^TA^TQ^*c^*+c^{*T}Q^*c^*)\le 1 \\
\iff & \min_{\lambda,\beta\in\mathbb{R}}\left\{\beta | \lambda \ge 0 \mbox{ and }
	\lambda F(Q,c)+\beta E_{n+1} \succeq G \right\} \le 0 \\
	&\mbox{\quad\quad\quad\quad where } G =
	\begin{pmatrix}
	A^TQ^*A & -A^TQ^*c^* \\
	(-A^TQ^*c^*)^T ~~& c^{*T}Q^*c^*-1
	\end{pmatrix}
\end{aligned}
\end{equation}

We can hence verify the inclusion by finding suitable parameters and
verifying the resulting LMI.

We finally have to perform the sound computation of the center translation
$(c+b)$. Again,
this is computed in floating-point arithmetic and we may have to
increase the
ratio of $Q$ to ensure the verification of the inclusion condition in
interval arithmetics: in the test of Theorem~\ref{Sproc}, we first
compute $(c+b)$ and $F(Q,c+b)$ in floating-point arithmetic (and possibly
increase the ratio), and with the
LMI we check that it ``contains'' $F(Q,c+b)$ directly computed in
interval arithmetic.

\subsection{Variable Packing}
It can be useful to analyze groups of variables independently, and merge
the results.
Given a set of variables
$\{x_1,\ldots, x_p, x_{p+1},\ldots,x_{p+q}\}$ and an ellipsoidal
constraint over these variables $(Q,c)$, we can find an ellipsoidal
constraint linking $x_1,\ldots,x_p$ by computing the assignment
defined by the matrix
$
\begin{pmatrix}
I_p & 0\\
0 & 0
\end{pmatrix}
$.

Given two sets of variables  $\{x_1,\ldots, x_p\}$ and
$\{ x_{p+1},\ldots,x_{p+q}\}$ linked respectively by $(Q_1, c_1)$
and $(Q_2, c_2)$, their product is tightly overapproximated by:
$
\Ell\left (
\begin{pmatrix}
\frac{Q_1}{2} & 0 \\
0 & \frac{Q_2}{2}
\end{pmatrix} ,
\begin{pmatrix}
c_1 \\
c_2
\end{pmatrix}
\right )
$.

\section{Verifying Linear Matrix Inequalities}
\label{LMI}

We now describe how we check the LMI's that determine the soundness of our
analysis. We use interval arithmetic: the coefficients are
intervals of floating-point numbers. Each atomic operation (addition,
multiplication\ldots) is overapproximated in the interval domain.

\subsection{Cholesky Decomposition}

Recall that the SDP solver  gives us an inequality of
the form
$
\sum\limits_{i=0}^r \alpha_iA_i
\succeq 0
$,
and candidate coefficients $(\alpha_i)$.

We translate each matrix and coefficient
into the interval domain, and sum them up in interval arithmetic
so that the soundness of the result does not depend on the floating-point
computation of the linear expression.

Then, we compute the Cholesky decomposition of the resulting matrix in
interval arithmetic. That is, we decompose \cite{mpmath}
 the matrix $A$ into
$A=LDL^T$ where $D$ is an interval diagonal matrix and $L$ a
(non interval) lower triangular
matrix with ones on the diagonal. Checking that $D$ has only positive
coefficients implies that $A$ is definite positive.

	\subsection{Practical Aspects of the Ellipsoidal Operations}
\subsubsection{The Precision Issue.}
The limitation in the precision of the computations makes us unable
to actually test whether a matrix is semidefinite positive: we can
only decide when a matrix is definite positive ``enough''. For
instance, standard libraries\footnote{E.g mpmath\cite{mpmath}} fail
at deciding that the null matrix $0$ is semidefinite positive.

It means that for all the operations and
verifications, we have to perform additional overapproximations
in addition to
those made by the SDP solver, such as multiplying the
ratio of the ellipsoid by a number $(1+\epsilon)$. Moreover, in the
verification of LMI's, it can prove useful to explore the
neighborhood $\alpha_i\pm\epsilon>0$ of the parameters $(\alpha_i)$.
As Roux and Garoche write in \cite{plg2}:
``Finding a good way to pad equations to get correct results, while
still preserving the best accuracy, however remains some kind of
black magic.''

\subsubsection{Complexity Results.}
From the complexity results of Porkolab and Khachiyan,
the resolution of an LMI with $m$ terms in dimension $n$ has a
complexity of $O(mn^4)+n^{O(\min(m,n^2))}$ \cite{complex}. Hence the complexity
of the abstract operations is polynomial as a function of the dimension $n$
(i.e., the number of variables), with a degree almost always smaller
than 4 (for most operations, $m\le4$).
The complexity of the Cholesky decomposition can be directly
computed and is $O(n^3)$.

\section{Conic Extrapolation}
\label{conop}

Now that the ellipsoidal operations are well defined, we can describe
the construction of the conic extrapolation. The goal is to analyze
variable transformations when the ellipsoidal radius evolves
(sub-)linearly in the value of the loop counters.

Let us have numerical variables $x=(x_1,\ldots, x_n)$ and loop counters
$y=(y_1,\ldots,y_k)$, we want to control the evolution of
$x$ depending on the counters $y$, which are expected to be monotonically increasing.

Inspired by the ellipsoidal constraints, we can use intersections
of linear inequalities and a quadratic constraint of the form:
\begin{definition}[Conic extrapolation]
\label{defcon}
Let $q$ be a definite positive quadratic form (that is, there is a matrix
$Q\succ 0$ such that $\forall x\in\mathbb{R}^n$, $q(x)=x^TQx$),
$c\in\mathbb{R}^n$, and for $i\in\llbracket 1,k \rrbracket$, $\beta_i>0$,
$\delta_i\in\mathbb{R}^n$, $\lambda_i\in\mathbb{R}$, and $b_i$ a boolean value.
We define the ellipsoidal cone:
\begin{multline*}
\on{Con}((q,c), (\beta_i,\delta_i,\lambda_i, b_i)_{1\le i\le k} ) =\\
\{ (x,y)\in \mathbb{R}^n\times\mathbb{R}^k |
\forall i\in \segk, y_i\ge\lambda_i\mbox{ }\land \\
\forall i\in \segk, (b_i\lor(y_i = \lambda_i))
\mbox{ }\land\\
 q(x-c-\Sk(y_i-\lambda_i)\delta_i)
\le (\Sk\beta_i(y_i-\lambda_i) +1)^2 \}
\end{multline*}
\end{definition}

Let $Q$ be the matrix associated with $q$, $\Ell(Q,c)$ is the ellipsoidal base of
the cone.
The $\lambda_i\in\mathbb{R}$ are the base levels of the cone, that is
the minimum values of the loop counters (usually zero). The
$\delta_i\in\mathbb{R}^n$ are the directions toward which the cone
is ``leaning''
(for instance, with a single loop iterating $\on{x\gets x+1}$, we
would want $\delta$ to be equal to $1$). The $\beta_i\in\mathbb{R}$
determine
the slope of the cone in each dimension. The $b_i$ are boolean
values stating, for each dimension, whether an extrapolation has been
made in this dimension. That is, do we consider only the $(x,y)$ with
$y_i=\lambda_i$ (case $b_i=\on{False}$) or all those with $y_i\ge\lambda_i$ and
verifying the other conditions (case $b_i=\on{True}$).

\begin{figure}[htbp]
\center
\includegraphics[width=0.7\textwidth]{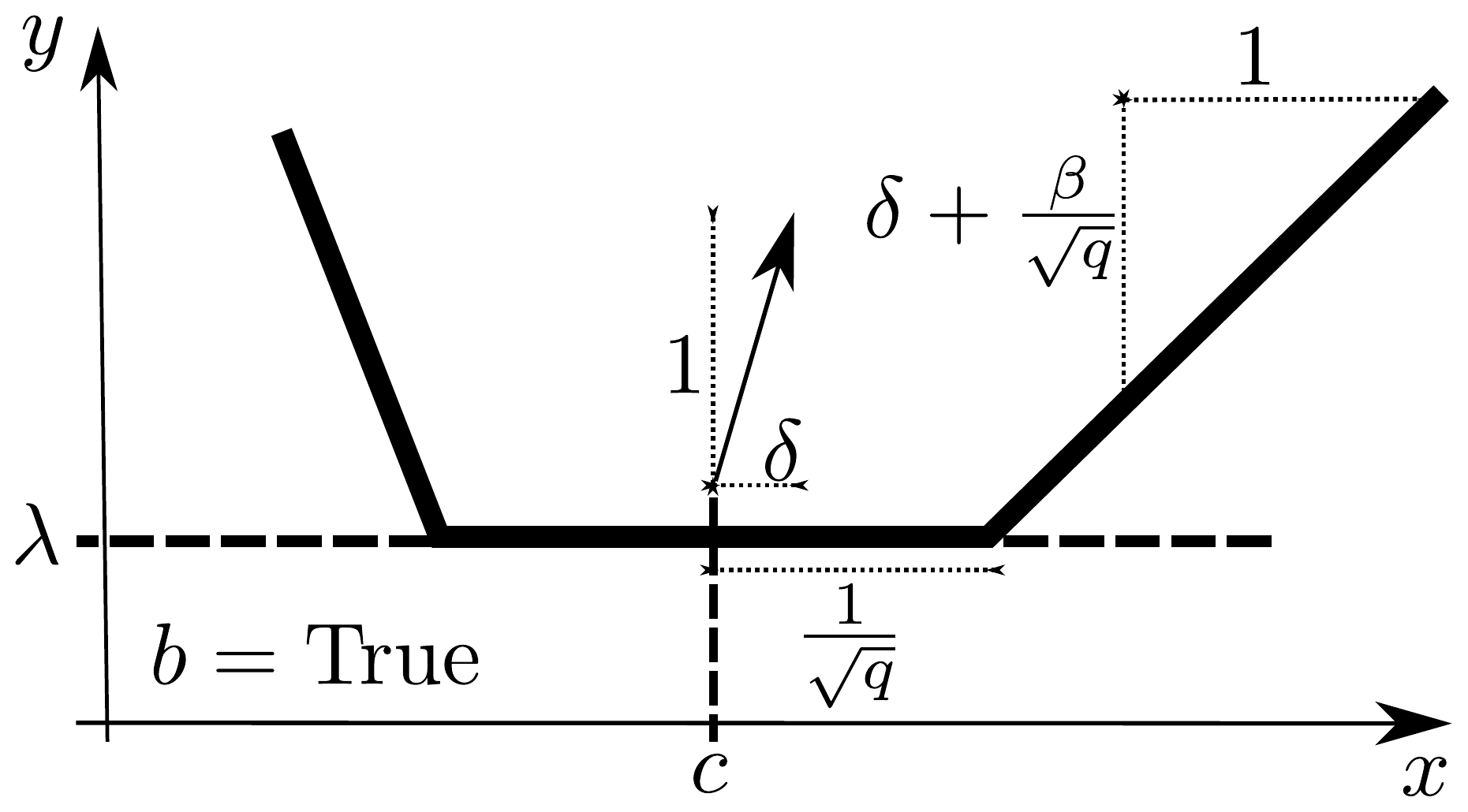}
\caption{Example for $p=1, k=1$.}
\end{figure}

\subsection{Conditions of Inclusion}
We need to be able to test the inclusion of two cones.
The following theorem shows that this inclusion can be reframed as
conditions that can be verified with an SDP solver.
\begin{theorem}
\label{thincl}
If we consider two cones $C=\on{Con}((q,c),
(\beta_i,\delta_i,\lambda_i, b_i)_{1\le i\le k})$ and
 $C'=\on{Con}((q',c'),
(\beta_i',\delta_i',\lambda_i', b_i')_{1\le i\le k})$,
then $C\subset C'$ if and only if
$$
\left\{
\begin{aligned}
& (i)\quad\forall i\in\segk, \lambda_i'\le\lambda_i
			\mbox{ and } \lambda_i>\lambda_i' \Rightarrow b_i'\\
& (ii)\quad\Ell(q,c)\subset \Ell\left(\frac{q'}{
		(1+\Sk\beta_i'(\lambda_i-\lambda_i'))^2},
		c'+\Sk(\lambda_i-\lambda_i')\delta_i'\right)\\
& (iii)\quad\forall i\in\segk ,
	b_i \Rightarrow \left(b_i' \mbox{ and } {\beta_i'}^2\ge
	\max_{u\in\mathbb{R}^n,q(u)\le 1}q'(\beta_iu+\delta_i-\delta_i') \right) \\
\end{aligned}
\right.
$$
\end{theorem}

To prove this theorem, we first consider the case when the two cones have the
same base levels. That is, we reduce it to the case when all the $\lambda_i$'s
are equal to $0$.

\begin{lemma}
\label{lemm1}
If we consider two cones $C=\on{Con}((q,c),
(\beta_i,\delta_i,0, b_i)_{1\le i\le k})$ and
$C'=\on{Con}((q',c'),
(\beta_i',\delta_i',0, b_i')_{1\le i\le k})$,
then $C\subset C'$ if and only if
$$
\left\{
\begin{aligned}
& (i)\quad\Ell(q,c)\subset \Ell(q',c') \\
& (ii)\quad\forall i\in\segk ,
	b_i \Rightarrow \left(b_i' \mbox{ and } {\beta_i'}^2\ge
	\max_{u\in\mathbb{R}^n,q(u)\le 1}q'(\beta_iu+\delta_i-\delta_i') \right) \\
\end{aligned}
\right.
$$
\end{lemma}

The proof of these two results is postponed to the end of the section.

\subsection{Test of Inclusion}
Theorem~\ref{thincl} enables us to build a sound test of
inclusion between two cones $C$ and $C'$. Condition $(i)$ can be
directly tested. We can use the procedure of the previous section
to test the ellipsoidal inclusion of condition $(ii)$.

Note that in practice, the test of inclusion will be used (during
widening iterations) on cones with the same ellipsoidal base. In
these cases, we do not want the overapproximations of the SDP solver
to reject the
inclusion. Therefore we should directly test whether the bases
$(Q,c)$ and the $\lambda_i$'s are equal (as numerical values) and answer
$\on{True}$ for
the test of base inclusion in this case.

To perform a sound test on the subcondition of $(iii)$:
$${\beta_i'}^2\ge\max_{u\in\mathbb{R}^n,q(u)\le 1}q'(\beta_iu+\delta_i-\delta_i')$$
we can compute an overapproximation
of
$M=\max_{u\in\mathbb{R}^n, q(u)\le 1}q'(\beta_iu+\delta_i-\delta_i')$

From Theorem~\ref{Sproc}, we know that
$$
M=1+\min_{s,t\in \mathbb{R}}\left\{t \mbox{ $|$ }
	 s\ge 0 \text{ and }  s F(\frac{Q}{\beta_i^2},0)+t E_{n+1}
	\succeq F(Q',\delta_i'-\delta_i)\right\}
$$
So for any feasible solution $(s,t)$ of this SDP problem, $1+t$ is
a sound overapproximation of $M$.

\subsection{Affine Operations on Cones}
\subsubsection{Counter Increment.}
The abstract counterpart of a statement $y_i\gets y_i+v$ for some
value $v$, is the operation $\lambda_i\gets \lambda_i+v$:
after the statement, the constraint is verified for $y_i-v$, and
making this change in Definition~\ref{defcon} leads to the new
value of $\lambda_i$.
To have a sound result, we can compute the sum in interval arithmetic,
and take the lower bound. Note that, in general, the loop counters
are integer valued. In that case, the value can be computed
exactly.

\subsubsection{Affine Transformations.}
We want to have a sound counterpart for the affine assignment
$x\gets Ax+b$ where $A$ is a matrix and $b$ a vector. Let us fix
the values of $(y_i)_{1\le i\le k}$ and note
$R=(1+\Sk\beta_i(y_i-\lambda_i))$ and
$\hat{c}=c+\Sk (y_i-\lambda_i)\delta_i$. Let $Q$\\ be the matrix of
$q$. We first want to find $(Q',c') \text{ such that we have the inclusion }\\
\{Ax+b | x\in\Ell(\frac{Q}{R^2},\hat{c})\} \subset
\Ell(\frac{Q'}{R^2},c')
$.
By symmetry, we can set $c'=A\hat{c}+b$. Thus, by doing the same
calculations as in Sect.\ref{ellaff}, we have
$$
\{Ax+b | x\in\Ell(\frac{Q}{R^2},\hat{c})\}
				\subset \Ell(\frac{Q'}{R^2},c')
\iff  Q\succeq A^TQ'A
$$
The last condition does not depend on the $y_i$'s, so for any
quadratic form $q'$ whose matrix $Q'$ verifies $Q\succeq A^TQ'A$
(which is an SDP equation, we can add the conditions of volume
minimization of (\ref{sdpU}), and
$Q'\preceq \frac{1}{\epsilon} I_n$ with
$\epsilon$ small enough, to ensure numerical convergence),
we have
\begin{multline*}
\{ (Ax, y) | (x,y)\in
\on{Con}((q,c), (\beta_i,\delta_i,\lambda_i, b_i)_{1\le i\le k} ) \}
\subset \\
\on{Con}((q',Ac),(\beta_i,A\delta_i,\lambda_i, b_i)_{1\le i\le k} )
\end{multline*}

As in the case of ellipsoidal assignments, $Ac+b$ and $A\delta_i$ are
computed in floating-point arithmetic. Hence once they are computed,
we have to ensure that the resulting numerical cone contains the
formally defined cone, i.e. we need to verify the
inclusion of
ellipsoidal bases with (\ref{verifassign}) and the procedure
described in Sect.\ref{ellaff}. We also
need to verify the
conic inclusion, i.e. the fact that
$\beta_i'\ge \beta_i+\sqrt{q'(A\delta_i-\delta_i')}$,
where $\beta_i'$ and $\delta_i'$ are the parameters of the resulting
cone.
So we may have to update the parameters and verify the
inequalities in a sound manner.

\subsection{Addition and Removal of Counters}
When the analyzer enters a new loop, it needs to take into account
the previous constraint and add a dependency on the current loop
counter $y_i$. Moreover, when it exits a loop, it needs to build a
new constraint overapproximating the previous one that does not
involve the counter $y_i$.

Ellipsoidal constraints can be seen as conic constraints with $k=0$.
Hence we study the problem of adding and removing counters to a
conic constraint
$\on{Con}((q,c), (\beta_i,\delta_i,\lambda_i, b_i)_{1\le i\le k} ) $.

\subsubsection{Adding a Counter.}

Let $y_{k+1}$ be the counter we want to add. Let $\lambda_{k+1}$ be
the minimal value of the counter inferred at this point. We set
$\beta_{k+1}=0$, $\delta_{k+1}=0$ and $b_{k+1} = \on{True}$ if the
value of $y_{k+1}$ at this point of the analysis is not known precisely,
 else if we know that $y_{k+1}=\lambda_{k+1}$, then $b_{k+1}=\on{False}$.

That gives us the constraint
$\on{Con}((q,c), (\beta_i,\delta_i,\lambda_i, b_i)_{1\le i\le k+1})$.

\begin{proof}
It is immediate from Definition~\ref{defcon}
and the distinction made on what
we know about $y_i$, that this constraint overapproximates the set of
reachable $(x,y)$ at this point.
\qed
\end{proof}

\subsubsection{Removing a Counter.}
We now want to remove the counter $y_k$ from the conic constraint,
provided that we know that $y_k\in[\lambda_k, M]$ with $M<+\infty$
(note that if it happens that $b_k=\on{False}$, then $M=\lambda_k$).

\begin{figure}[h]
\center
\includegraphics[width=0.5\textwidth]{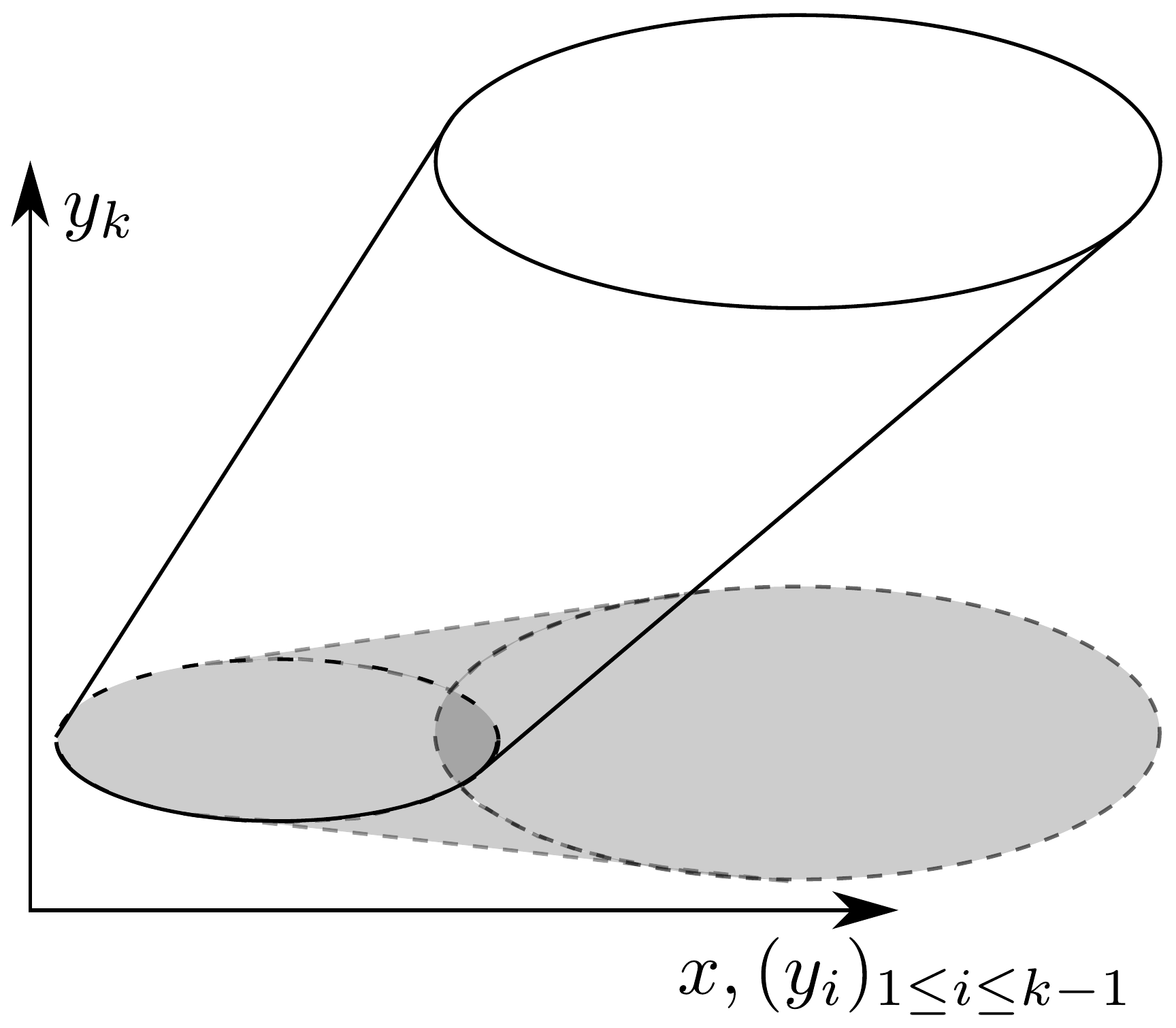}
\caption{Removing the counter $y_k$, hence projecting along its
direction.}
\end{figure}

\begin{theorem}
\label{proj}
Let $C=\on{Con}((q,c),
(\beta_i,\delta_i,\lambda_i, b_i)_{1\le i\le k+1})$. Then $C$ is
convex and we have
$$C_{|y_k\in[a,b]}=C\cap\{(x,y) | y_k\in[a,b]\}= \on{Conv}(
C\cap\{(x,y) | y_k=a \lor y_k=b\}).$$
where we suppose $a\ge\lambda_k$ and where $\on{Conv}(X)$ is the
convex hull of $X$.
\end{theorem}
\begin{proof}
Up to translation, we can assume that
$\forall i\in\segk, \lambda_i=0$.
If $a=b$, it is immediate. We suppose $a<b$. Then, if $Q$ is the
matrix of $q$, let $S$ be the inverse of its square root
($S^{-2}=Q$). We have
$$
\begin{aligned}
(x,y)\in C  \iff &  q(x-c-\Sk y_i\delta_i)
		\le(1+\Sk\beta_iy_i)^2 \\
		\iff & \exists u\in\mathbb{R}^p, ||u||_2\le 1,
		x-c-\Sk y_i\delta_i =
		(1+\Sk\beta_iy_i)Su 		\\
		\iff & \exists u\in\mathbb{R}^p, ||u||_2\le 1,
	 x = \frac{y_k-a}{b-a}x_b + (1-\frac{y_k-a}{b-a})x_a \\
	 \mbox{where } &
	 x_b = (c+\sum\limits_{i=1}^{k-1}y_i\delta_i+
		(1+\sum\limits_{i=1}^{k-1}\beta_iy_i)Su +
		b(\delta_k+\beta_kSu)) \\
	& x_a =(c+\sum\limits_{i=1}^{k-1}y_i\delta_i+
		(1+\sum\limits_{i=1}^{k-1}\beta_iy_i)Su +
			a(\delta_k+\beta_kSu))
\end{aligned}
$$
From the
previous equivalences, we have
$z_a=(x_a,y_1,\ldots,y_{k-1},a)\in C$ and
$z_b=(x_b,y_1,\ldots,y_{k-1},b)\in C$. Moreover
$(x,y)=\frac{y_k-a}{b-a}z_b + (1-\frac{y_k-a}{b-a})z_a$.
\qed
\end{proof}
Let $\pi_{y_k}$ be the projection along $y_k$. Since convexity and
barycenters are preserved up to projections,
$
\pi_{y_k}(C_{|y_k\in[a,b]})=\on{Conv}(
\pi_{y_k}(C_{|y_k=a}) \cup
\pi_{y_k}(C_{|y_k=b}))
$. So, by a direct calculation
$$
\begin{aligned}
\pi_{y_k}(C_{|y_k=a}) &=
 \on{Con}((\frac{q}{(1+\beta_k(a-\lambda_k))^2},
				c+(a-\lambda_k)\delta_k),\\
					&\mbox{\quad\quad\quad\quad\quad\quad\quad\quad}
(\frac{\beta_i}{1+(a-\lambda_k)\beta_k},
	\delta_i,\lambda_i,b_i)_{1\le i\le k-1})
\end{aligned}
$$
We have a similar equality for $\pi_{y_k}(C_{|y_k=b})$, hence we just
have to compute the join
$(\pi_{y_k}(C_{|y_k=a}) \bigsqcup_{\on{Con}}
\pi_{y_k}(C_{|y_k=b}))$, which is an
overapproximation of the convex hull of the union.

However, we have to implement this operation such that it is sound
when computed in floating-point arithmetic. Via affine transformation,
we can soundly compute an
ellipsoidal base $\Ell(q^*,c^*)$ such that
$$
\left\{
	  \begin{aligned}
	  & q(x-(c+(a-\lambda_k)\delta_k))\le (1+\beta_k(a-\lambda_k))^2\\
	  & q(x-(c+(b-\lambda_k)\delta_k))\le (1+\beta_k(b-\lambda_k))^2\\
	  \end{aligned}
		  \right.
\Rightarrow q^*(x-c^*)\le 1
$$
Then, for any $i\in\llbracket1,k-1\rrbracket$ such that
$b_i=\on{True}$, if we note $\beta_i^*$ and $\delta_i^*$ the
parameters of the resulting cone, in order to have an inclusion of
$C_{|y_k=a}$ and $C_{|y_k=b}$ in $C^*$, we need to establish by
Theorem~\ref{thincl} that:
$$
\left\{
\begin{aligned}
& \beta_i^{*2} \ge\max_{\frac{q(u)}{(1+\beta_k(a-\lambda_k))^2}\le 1}
   q^*(\frac{\beta_i}{1+(a-\lambda_k)\beta_k}u+\delta_i-\delta_i^*)\\
& \beta_i^{*2} \ge\max_{\frac{q(u)}{(1+\beta_k(b-\lambda_k))^2}\le 1}
   q^*(\frac{\beta_i}{1+(b-\lambda_k)\beta_k}u+\delta_i-\delta_i^*)
\end{aligned}
	\right.
$$
And since from our hypothesis on $\Ell(q^*,c^*)$ we know that
$q\succeq q^*$, we can just set $\delta_i^*=\delta_i$ and the
condition becomes $\beta_i^*\ge\beta_i$. So we can just define
$\beta_i^*=\beta_i$, hence the resulting cone after the removing of
the $k^{\mbox{th}}$ counter is
$$
\on{Con}((q^*,c^*),(\beta_i,\delta_i,\lambda_i,b_i)_{1\le i\le k-1})
$$

\subsection{A Widening Operator}
Let $C=
	\on{Con}((q,c), (\beta_i,\delta_i,\lambda_i, b_i)_{1\le i\le k})$
and $C'=
	\on{Con}((q',c'),
		(\beta_i',\delta_i',\lambda_i', b_i')_{1\le i\le k})$.
We suppose that
$\forall i\in\segk, \lambda_i\le\lambda_i'$ and
$\exists i\in\segk,\lambda_i<\lambda_i'$.

We want to define a widening operator $\W$ over cones. The intuitive
idea is
that if $C$ ``starts strictly below'' $C'$ (cf. the conditions on
the $\lambda_i$'s), then $C^*=C\W C'$ has the same ellipsoidal base
as $C$, but its opening has been ``widened'' to contain $C'$. The
decision of only changing the opening and the orientation of the cone
(i.e., to change only the $\beta_i$'s and $\delta_i$'s) relies on the
hypothesis that the relative shift of $C'$ from $C$ has good
chances to be reproduced again. Hence the name of
``conic extrapolation''.

\subsubsection{Definition of $\W_p$.}

More formally, we first study the special case in which we know
that $\Ell(q',c')\subset C$ and we define a partial widening
operator $\W_p$.

Let $C^*=C\W_pC'= \on{Con}((q,c),
(\beta_i^*,\delta_i^*,\lambda_i, b_i^*)_{1\le i\le k})$.
We note $(i),(ii),(iii)$ (resp. $(i'),(ii'),(iii')$) the
conditions of Theorem~\ref{thincl}  relative to the inclusion
$C\subset C^*$ (resp. $C'\subset C^*$). By construction of $C^*$,
we already have $(ii)$ and with our hypothesis
$\Ell(q',c')\subset C'$,
we just need to verify $C\subset C^*$ to have $(ii')$.
We can define $\forall i \in\segk,
b_i^*=(b_i\lor b_i'\lor \lambda_i<\lambda_i')$, which gives us $(i)$
and $(i')$.

Finally to verify $(iii)$ and $(iii')$, we only need to define the
$\beta_i^*$'s and $\delta_i^*$'s such that
$$
\forall i\in\segk,
\left\{
\begin{aligned}
& b_i \Rightarrow  {\beta_i^*}^2\ge
	\max_{q(u)\le 1}q(\beta_iu+\delta_i-\delta_i^*)  \\
& b_i' \Rightarrow  {\beta_i^*}^2\ge
	\max_{q'(u)\le 1}q(\beta_i'u+\delta_i'-\delta_i^*)  \\
\end{aligned}
\right.
$$
which we overapproximate by a triangle inequality for the norm defined by $q$:
$$
\forall i\in\segk,
\left\{
\begin{aligned}
& b_i \Rightarrow  \beta_i^*\ge
	\beta_i+\sqrt{q(\delta_i-\delta_i^*) } \\
& b_i' \Rightarrow  \beta_i^*\ge
	\beta_i'r+\sqrt{q(\delta_i'-\delta_i^*) } \\
& \quad\quad\quad\mbox{ where }
		r\ge\min\{\rho>0 | \frac{q'}{\rho^2}\preceq q\}
\end{aligned}
\right.
$$

With the SDP methods of the first section, we can compute an
overapproximating $r$.
For each $i$, if none of $b_i$ or $b_i'$ is $\on{True}$, then from
our
hypothesis $\Ell(q',c')\subset C'$, $b_i^*=\on{False}$ and we do not
have
to give values to either $\beta_i$ or $\delta_i$. If only $b_i$ (resp.
$b_i'$) is $\on{True}$, then we define
$\delta_i^*=\delta_i$ and $\beta_i^*=\beta_i$
(resp. $\delta_i^*=\delta_i'$ and $\beta_i^*\ge r\beta_i'$).

If $b_i=b_i'=\on{True}$, we want to minimize
$
\max(
	\beta_i+\sqrt{q(\delta_i-\delta_i^*) } ,
	\beta_i'r+\sqrt{q(\delta_i'-\delta_i^*) } ).$
If we fix the $q$-distance $\sqrt{q(\delta_i-\delta_i^*)}$, we want to
minimize the $q$-distance $\sqrt{q(\delta_i'-\delta_i^*)}$. With this
geometrical point of view, we see that the optimal $\delta_i^*$ is
a barycenter of $\delta_i$ and $\delta_i'$.

So we define $\delta_i^*=\mu\delta_i+(1-\mu)\delta_i'$ where we want
to find $\mu\in[0,1]$ minimizing
\begin{multline*}
\max(
	\beta_i+\sqrt{q(\delta_i-\mu\delta_i-(1-\mu)\delta_i') } ,
	\beta_i'r+\sqrt{q(\delta_i'-\mu\delta_i-(1-\mu)\delta_i') } )
= \\
\max(
	\beta_i+(1-\mu)\sqrt{q(\delta_i-\delta_i')} ,
	\beta_i'r+\mu\sqrt{q(\delta_i-\delta_i')} ).
\end{multline*}
Hence, we can exactly (up to floating-point approximations) compute $\mu$, define
$\delta_i^*$ and then $\beta_i^*$.
This construction of $(\beta_i^*,\delta_i^*,b_i^*)_{1\le i\le k}$
ensures that $C,C'\subset C^*$ and defines $\W_p$.

\subsubsection{Definition of $\W$.}

We now study the general case in which the only assumption made is
that
$\forall i\in\segk, \lambda_i\le\lambda_i'$ and
$\exists i\in\segk,\lambda_i<\lambda_i'$.

We define a cone
$C^+ =
\on{Con}((q,c),
	(\beta_i^+,\delta_i^+,\lambda_i, b_i^+)_{1\le i\le k} )$,
which contains the ellipsoidal
base of the two cones. The definition (with $q$, $c$ and the $\lambda_i$)
ensures the inclusion of the
ellipsoidal base of $C$.
Let $R=1+\Sk\beta_i^+(\lambda_i'-\lambda_i)$
and $\Delta=\Sk\delta_i^+(\lambda_i'-\lambda_i)$.

We define $\forall i\in\segk, b_i^+=(\lambda_i < \lambda_i')$.
Let $ r=\min\{\rho>0 | \frac{q}{\rho^2}\preceq q'\}$, if we have
$\Ell(\frac{q}{r^2},c')\subset\Ell(\frac{q}{R^2}, c+\Delta)$ then
$\Ell(q',c')\subset\Ell(\frac{q}{R^2}, c+\Delta)$ and from the
definitions of the $b_i^+$'s and Definition~\ref{defcon},
we would have
$\Ell(q',c')\times\{(\lambda_1',\ldots,\lambda_k')\}\subset C^+$.
To get this result, we need:
$$
 \Ell(\frac{q}{r^2},c')\subset\Ell(\frac{q}{R^2}, c+\Delta)
 \iff \{q(c'-c-\Delta) \le (R-r)^2\} \land \{R\ge r\}
$$
If the transformations applied to the cone are affine, the shift can be
seen as the difference between centers.
So we choose to define $\Delta=c'-c$.
Then we choose the minimal possible value of $R$ to have a cone as
tight as possible: once the  $\delta_i$'s corresponding to $\Delta$
are computed in floating-point arithmetic, we can define an upper bound on
$\sqrt{q(c'-c-\Delta)}+r$ and define $R$ accordingly, so that the above
inequality is verified.

These definitions of $\Delta$ and $R$ must be implemented in terms of
$\beta_i^+$ and $\delta_i^+$. Since we ensured that there is at least
one $i$ such that $(\lambda_i'-\lambda_i)\neq 0$, there is always a
solution. If only one $i$ fits this criterion the solution is unique,
otherwise a choice must be made on how to weight the different
variable.

This uncertainty can be easily explained: recall that in real
programs, only one loop counter is increased at a time, so we know
what causes the change in our constraint. This is not the case if
many loop counters are increased at the same time.

Finally, this definition of $C^+$ allows us to define the widening
operator $\W$ by:
$C\W C'=\left(C\W_pC^+\right)\W_pC'$. Note that the assumptions of
$\W_p$
are verified since $C$ and $C^+$ have the same ellipsoidal base, and
$C^+$, hence $C\W_pC^+$, contains the base of $C'$.

To ensure the convergence of the widening sequence in the cases described
in Sect.~\ref{switch}, we can use a real widening operator on the
$\beta_i$'s that sets them to $+\infty$ after a certain number of steps,
for instance.

\begin{figure}[hbpt]
\center
\includegraphics[width=0.7\textwidth]{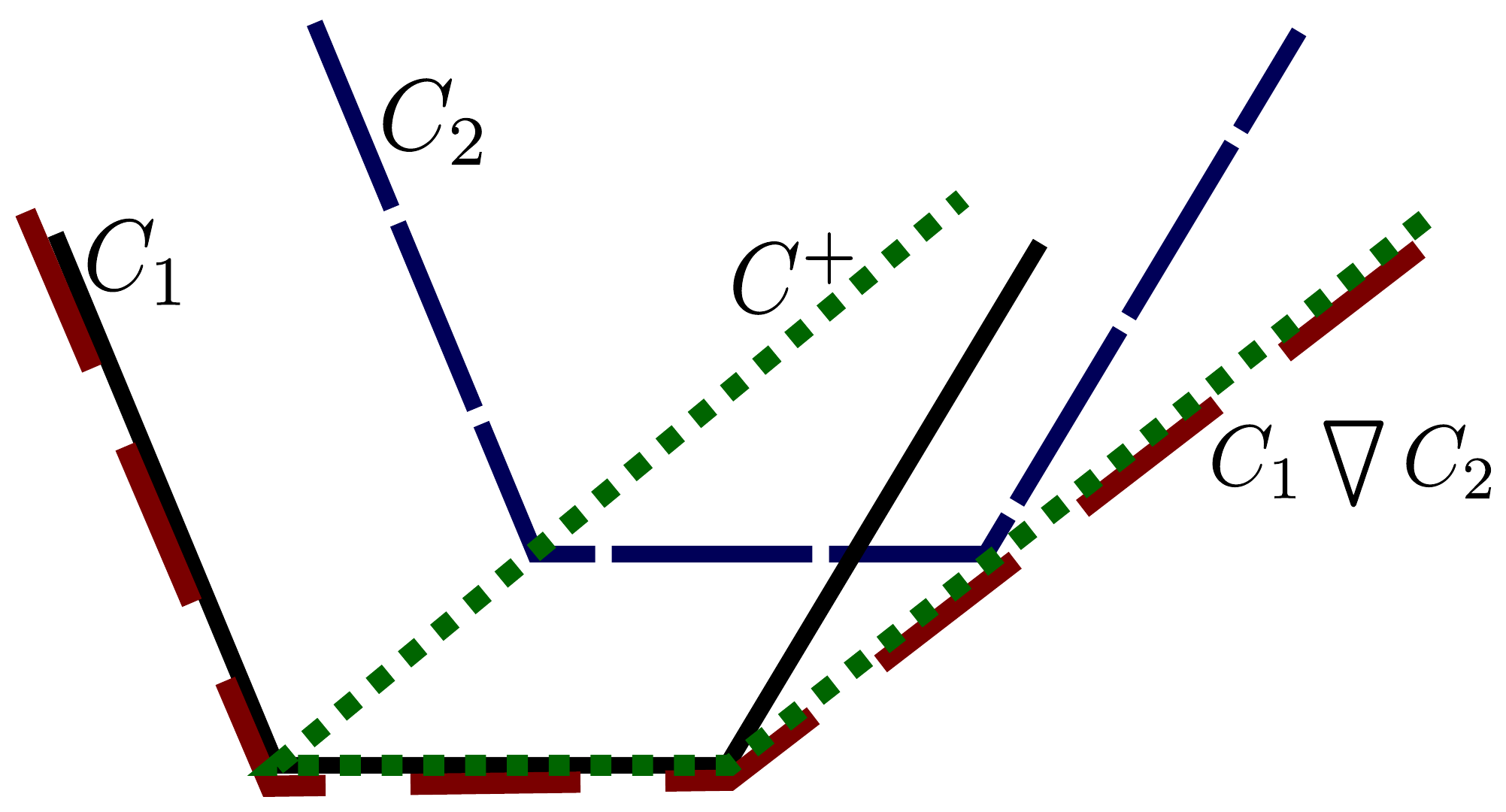}
\caption{Example showing the various cones involved in the
definition of $\W$: $C_1$ in black, $C_2$ in blue, $C^+$ in green and $C_1\W C_2$ in red.}
\end{figure}

\subsection{Proof of the Characterization of Conic Inclusion}

\begin{proof}[Proof of Lemma~\ref{lemm1}]
\quad\newline
$\bullet$ We first prove that $(i)\land(ii)\Rightarrow C\subset C'$.
From $(i)$, we know that \newline
$
\forall u \in\mathbb{R}^p, q(u)\le 1 \Rightarrow
 q'(u+c-c')\le 1
$.
Thus, for $(x,y)\in\mathbb{R}^{p}\times\mathbb{R}_+^k$ such that
$
q(x-c-\Sk y_i\delta_i )\le(1+\Sk \beta_iy_i)^2
$,
we define $\nu=(1+\Sk\beta_iy_i) \ge \sqrt{q(x-c-\Sk y_i\delta_i)}$
and $u=\frac{1}{\nu}(x-c-\Sk y_i\delta_i)$. We have $q(u)\le1$.
\begin{equation}
\label{ineq}
\begin{aligned}
q'(x-c'-\Sk y_i\delta_i')
& = q'\left(u+(c-c') + (\nu-1)u+\Sk y_i(\delta_i-\delta_i')\right)\\
& \le \left(\sqrt{q'(u+c-c')}
		+\Sk y_i\sqrt{q'(\beta_iu+\delta_i-\delta_i')}\right)^2 \\
\end{aligned}
\end{equation}
So if $b_i=\on{False}$ then $y_i=0$, hence
$ y_i\sqrt{q'(\beta_iu+\delta_i-\delta_i')} \le y_i\beta_i'$,
and from $(ii)$, if $b_i=\on{True}$, then we have the
same inequality since $q(u)\le 1$. So this inequality is true for all
$i\in\segk$.

Thus, we have $q'(x-c'-\Sk y_i\delta_i')\le (1+\Sk\beta_i' y_i)^2$
and $\forall i\in\segk$ we have $ y_i\ge 0$ and from $(ii)$,
$y_i>0 \Rightarrow b_i\Rightarrow b_i'$. So $(x,y)\in C'$.
So $C\subset C'$.

$\bullet$ Now we prove that $C\subset C'\Rightarrow (i)\land(ii)$.
It is obvious that $C\subset C'\Rightarrow (i)$ by taking the
intersections of the cones with the set
$\{(x,0)\in\mathbb{R}^{p+k}\}$.

If $\exists i\in\segk \st b_i'=\on{False}$ and $b_i=\on{True}$, then
there exist a point $(x,y)$ of $C$ with $y_i > 0$, so
$(x,y)\notin C'$ and $C\not\subset C'$.

If $\exists i\in\segk \st b_i=\on{True}$ and
$\beta_i'<\max_{q(u)\le 1}\sqrt{q'(\beta_iu+\delta_i-\delta_i')}$,
then let us take $u\in\mathbb{R}^p$ such that $q(u)\le 1$ and
$\beta_i'<\sqrt{q'(\beta_iu+\delta_i-\delta_i')}$. We define
$x(t)=(1+\beta_it)u+t\delta_i+c$.

For any $t \ge 0$,
$
q'(x(t)-c'-t\delta_i') \ge \left(\sqrt{q'(u+c-c')}
		- t\sqrt{q'(\beta_iu+\delta_i-\delta_i')}\right)^2 \\
$

Since $\beta_i'<\sqrt{q'(\beta_iu+\delta_i-\delta_i')}$, and from
the previous inequality, we have for $t$ big enough,
 $q'(x(t)-t\delta_i'-c') > (1+\beta_i')^2$.
By a direct computation $q(x(t)-t\delta_i-c)\le(1+\beta_it)^2$.
And since
$b_i=\on{True}$, we have for $y(t)$ such that $y(t)_i=t$ and
$y(t)_{j\neq i}=0$, $(x(t),y(t))\in C$, $(x(t),y(t))\notin C'$ so
$C\not\subset C'$.

We have proven $\neg(i)\lor\neg(ii)\Rightarrow C\not\subset C'$, so
we finally have $(i)\land(ii)\iff C\subset C'$.
\qed
\end{proof}

\begin{proof}[Proof of Theorem~\ref{thincl}]
Since $(c,
(\lambda_i)_{1\le i \le k})\in C$, it is clear that $(i)$ is a
necessary condition. Moreover,
$C\subset\{(x,y)\in\mathbb{R}^{p+k} | \forall i \in
\segk, y_i\ge\lambda_i\}=:\on{Orth}_\lambda$,
therefore $C\subset C' \iff C\subset C'\cap\on{Orth}_\lambda
\iff C\subset C'\cap\on{Orth}_\lambda\land (i)$.

Directly from Definition~\ref{defcon}, we have for
$R=(1+\Sk\beta_i'(\lambda_i-\lambda_i'))$ the implication
$
(i)\Rightarrow C'\cap\on{Orth}_\lambda =
\on{Con}((
\frac{q'}{R^2},
c'+\sum\limits_{i=1}^k(\lambda_i-\lambda_i')\delta_i'),
(\frac{\beta_i'}{R},\delta_i',\lambda_i, b_i')_{1\le i\le k} )
$.
So up to translation, we can reduce this case to $\lambda_i=\lambda_i'=0$
and apply Lemma~\ref{lemm1}.
$$
\begin{aligned}
 C\subset C'  \iff& C\subset C'\cap\on{Orth}_\lambda\land (i) \\
\iff	&  (i) \land
		   \Ell(q,c)\subset\Ell(\frac{q'}{R^2},c'+
		  				\Sk (\lambda_i-\lambda_i')\delta_i') \\
			&\land
		   \forall i\in\segk, b_i\Rightarrow \left( b_i' \mbox{ and }
		  		\frac{{\beta_i'}^2}{R^2}\ge
			\max_{q(u)\le 1} \frac{q'}{R^2}(\beta_iu+\delta_i-\delta_i')
		  		\right)\\
\iff& (i)\land(ii)\land(iii)
\end{aligned}
$$
\qed
\end{proof}

\section{Application and Convergence}
\label{switch}

\begin{figure}[h]
\begin{minipage}[c]{0.33\linewidth}
\begin{algorithmic}[1]
	\State $x\gets 0 \in \mathbb{R}^n$
	\For{$y$ from 0 to $\infty$}
		\State pick $i\in\llbracket 1,n\rrbracket$
		\State pick $\epsilon\in\{-1,1\}$
		\State $x_i\gets x_i+\epsilon$
	\EndFor
\end{algorithmic}
\end{minipage}
\begin{minipage}[c]{0.5\linewidth}
\begin{tabular}{|r|c|c|c|c|c|c|c|c|}
\hline
$n$ & 2 & 4 & 6 & 8 & 10 & 12 & 14 & 16 \\
\hline
Ell. cones& 3s & 7s & 19s & 49s & 1m56s & 4m16s &8m & 12m\\
\hline
Polyhedra & $<$0.1s & $<$0.1s & 0.3s & 2.5s & 54s & 47m &
			$>$1h & $>$1h \\
\hline
\end{tabular}
\end{minipage}
\caption[Execution test of ellipsoidal cones.]{Example of a
 program and its average analysis time on a \SI{2}{\GHz} CPU.
Ellipsoidal cones have been prototyped\footnotemark{} in Python using
NumPy\cite{numpy}, CVXOPT\cite{cvxopt}, mpmath\cite{mpmath}.
The Apron\cite{apron} C library has been used for polyhedra.}
\label{counters}
\end{figure}

\footnotetext{The prototype Python code and details about benchmarks are
available on \\ {http://www.eleves.ens.fr/home/oulamara/ellcones.html}}

In the definition of the conic extrapolation, we did not describe how
to choose the ellipsoidal base. For instance, it is possible to get
an ellipsoidal shape by computing some iterates of the loop. This
seems to work well in the case of
programs composed of loops and nondeterministic counter increments:
the iterations capture in which direction the counters
globally increase (Fig.~\ref{counters}). Since the diameter of the set
containing the numerical variables grows linearly, any cone will be
overapproximating for $\beta$ big enough.

\subsection{Switched Linear Systems}

\begin{figure}[h]
\begin{minipage}[c]{0.43\linewidth}
\begin{algorithmic}[1]
	\State $x\gets 0 \in \mathbb{R}^n$
	\State $(A_i,b_i)_{1\le i\le k}$ \\ \quad where
		$b_i\in\mathbb{R}^n, A_i\in\mathcal{M}_n(\mathbb{R})$
	\For{$y$ from 0 to $\infty$}
		\State $i\gets \operatorname{rand}(1,n)$
		\State $x\gets A_ix+b_i$
	\EndFor
\end{algorithmic}
\end{minipage}
\begin{minipage}[c]{0.67\linewidth}
\begin{tabular}{|r|c|c|c|c|c|c|c|}
\hline
Benchmark & 1 &2 &3 &4 &5 &6 &7  \\
\hline
Ell. Cones & 1s & 2s & 2s & 2s & 1.8s & 1.2s & 1.3s \\
\hline
Polyhedra & 3.2s & 16.6s & 18s & 24s & $>$1h & $>$1h & 2m35s \\
\hline
\end{tabular}
\end{minipage}

\caption{The structure of a Switched Linear System and some
benchmarks. Experimental conditions are the same as in
Fig.~\ref{counters}. Except for Benchmarks 1 and 2, the resulting
polyhedron is trivial.}
\label{figswitch}
\end{figure}

However, the picture is not as nice if we add linear transformation.
This is the case, for instance, for switched linear systems in
control
theory (Fig.\ref{figswitch}): if the
quadratic form associated with the ellipsoid is not a Lyapunov
function of the linear part of the system (of the
matrices $A_i$ in Fig.\ref{figswitch}), then the growth of its
radius is exponential in the loop counters, and cannot be captured
by our conic extrapolation. So if $Q$ is the matrix of the
ellipsoidal base of the cone, the Lyapunov conditions should be
verified simultaneously : $\forall i,Q-A_i^TQA_i\succeq 0$. This is
not always possible, but one can use the SDP solver to try to find a
suitable $Q$. Note that the identity matrix is not stable in the
sense of control theory and should not be included in the search of
 $Q$. When $Q$ verifies the Lyapunov conditions, it is easy to show
 that for $\beta_i$ large enough, the cone will be invariant during
 loops iterations.

\subsection{Proof of Convergence with the Lyapunov Condition}

Now we show the converse of the assertion of the previous
paragraph: if $q$ is a
Lyapunov function for the matrix $A$ and $b$ is a vector, then for
$\beta$ large enough $C=\on{Con}((q,c),(\beta,\delta,\lambda,\on{True}))$
is stabilized by the iteration of $x\gets Ax+b$.

\begin{proof}
Up to translation, we show the result for $\lambda=0$. We know than for
any $x\in C$, $q(x-c-y\delta)\le(\beta y+1)^2$ and by the Lyapunov
condition, there exist $\epsilon> 0$ such that
$\forall x, q(Ax)\le(1-\epsilon)q(x)$.

Let $\eta=1-\sqrt{1-\epsilon}$ and $M=(A-\on{Id})c+b-\delta$. We have
\begin{multline*}
\sqrt{q(Ax+b-c-(y+1)\delta)}=
		\sqrt{q(A(x-c-y\delta)+(A-\on{Id})(c+y\delta)+b-\delta)} \\
		\le \sqrt{1-\epsilon}(\beta y +1) +
				\sqrt{q((A-\on{Id})(c+y\delta)+b-\delta )} \\
		\le \beta y+1 +\sqrt{q(M+y(A-\on{Id})\delta)} - \eta\beta y \\
		\le \beta y + 1 + \sqrt{q(M)}+y(\sqrt{q((A-\on{Id})\delta)}-\eta\beta)
\end{multline*}

So for $\beta$ large enough ($\beta > \sqrt{q(M)}$ and
$\beta >\sqrt{q((A-\on{Id})\delta)}/\eta$), we have
$q(Ax+b-c-(y+1)\delta)\le(\beta(y+1)+1)^2$.
\qed
\end{proof}

\section{Concluding Remarks and Perspectives}
We proposed an abstract interpretation framework based on ellipsoidal
cones to study systems with (sub-)linear growth in loop counters.
The aim of this work is twofold: to build an extension of the formal
verification of linear systems, and to devise a framework that can be
used outside the context of digital filters. Indeed, only the choice
of the ellipsoidal base of the cone has to deal with control theory
considerations.

The next step is obviously to go beyond the prototype and have a
robust implementation to test this framework on actual systems.
This will involve a research on how to accurately
tune and use the SDP solver, how to deal with precision issues.

The main tools are the SDP solver and the SDP duality to check the
soundness of the results of operations via LMI's. However, we are not
bound to use SDP solvers to compute these results,
and exploring other options might speed up the analysis.

It would also be interesting to generalize this
framework to switched linear systems that are more complex than
those studied above. An example is the analysis in
\cite{switch}.

\subsubsection{Acknowledgments.}
We want to thank Pierre-Loïc Garoche and Léonard Blier for the
fruitful conversations we had with each of them during this work,
as well as the anonymous referees for the time and efforts taken
to review this work.

\end{document}